\documentclass[12pt,authoryear]{article}
\usepackage [english]{babel}
\usepackage [autostyle, english = american]{csquotes}
\MakeOuterQuote{"}
\usepackage{color}
\usepackage{amsfonts}
\usepackage{amsthm}
\usepackage{amssymb,bm}
\usepackage{amsmath,enumerate,rotate}
\usepackage{amssymb,latexsym,mathrsfs}
\usepackage{graphicx}

\RequirePackage[OT1]{fontenc}
\RequirePackage{amsthm,amsmath}
\RequirePackage[numbers]{natbib}
\RequirePackage[colorlinks,citecolor=blue,urlcolor=blue]{hyperref}
\usepackage[font=footnotesize]{caption}
\newcommand{\E}{\mathbb{E}}

\newcommand{\X}{\mathbb{X}}
\newcommand{\Y}{\mathbb{Y}}
\def\I{\mbox{I}}

\newcommand{\N}{\mathds{N}}

\def\Pn{\mbox{Pn}}
\def\E{\mbox{E}}
\def\V{\mbox{V}}
\newcommand{\Ga}{\Gamma}

\def\N{\mbox{N}}

\numberwithin{equation}{section}
\theoremstyle{plain}
\newtheorem{thm}{Theorem}[section]

\theoremstyle{definition}

\theoremstyle{remark}

\begin{document}

\title{Modelling and computation using NCoRM mixtures for density regression}

\author{Jim E. Griffin
and Fabrizio Leisen\footnote{{\it Corresponding author}: Jim E. Griffin, School of Mathematics, Statistics and Actuarial Science, University of Kent, Canterbury CT2 7NF, U.K. {\it Email}: 
jeg28@kent.ac.uk}\\
University of Kent}

\date{}



\maketitle

\abstract{
Normalized compound random measures are flexible nonparametric priors for related distributions. We consider building general nonparametric regression models using
normalized compound random measure mixture models. 
Posterior inference is made using a novel pseudo-marginal Metropolis-Hastings sampler for 
normalized compound random measure mixture models. The algorithm makes use of a new general approach to the unbiased estimation of Laplace functionals of compound random measures (which includes completely random measures as a special case). The approach is illustrated on problems of density regression.
}
\\
\noindent\textbf{Keyword}: Dependent random measures;  Mixture models; Multivariate L\'evy measures; Pseudo-marginal samplers; Poisson estimator.



\section{Introduction}

The problem of Bayesian nonparametric inference for distributions at different regressor values has been an extremely active area of research. Many approaches use dependent nonparametric mixture models and build on the idea of dependent Dirichlet process mixture models \citep{MacEachern}, which generalized the commonly-used Dirichlet process mixture model. A generic dependent nonparametric mixture model
assumes that a sample $y_1,\dots,y_n$ observed at regressor values $x_1,\dots,x_n$ (where $x\in\X$ for some measureable space $\X$)  is modelled as
\begin{equation}
y_i\vert x_i \sim q(y_i \vert \theta_{c_i}(x_i)),
\qquad p(c_i=k) = w_k(x_i),\quad k=1,\dots,\infty
\label{DDP}
\end{equation}
where $q(y\vert \theta)$ is a distribution for $y$
(where $y\in\Y$ for some measureable space $\Y$) 
 with parameter $\theta$, $w_k(x)\geq 0$ for all $k$ and $x\in\X$, $\sum_{k=1}^{\infty} w_k(x) = 1$ almost surely for all $x\in \X$ and $\theta_1(x), \theta_2(x), \theta_3(x), \dots$ are independent realisations of a stochastic process. We refer to $\theta_1(x), \theta_2(x), \theta_3(x), \dots$ as the locations of the mixture components. The model simplifies to a nonparametric mixture model if the sample is observed at a single regressor value. 

Many approaches to constructing specific models in the form of (\ref{DDP}) generalize the stick-breaking construction of the Dirichlet process \citep{sethuraman} and these were reviewed in \cite{dun10}. Alternatively,
models can be constructed by normalising dependent random measures. This generalizes the approach introduced by \cite{rlp} to an arbitrary dimension. These constructions  have several advantages. Firstly, the weights $w_1(x), w_2(x), \dots$ are not ordered, as is the case with many stick-breaking constructions. Secondly, dependence is defined at the level of the weights $w_k(x)$ rather than, as is typical in stick-breaking constructions, through a non-linear transformation of the weights. \cite{fotwil12} defined a wide-class of such process using normalized kernel-weighted random measures, which generalize the approach to time-dependent random measures in \cite{Gri11}. \cite{GKS13} developed an approach to modelling a finite set of dependent random measures using superpositions of completely random measure 
\citep[see also][]{NL14, NLP14a, NLP14b, Chen13}. Alternatively, dependence can be modelled through a L\'evy copula \citep{LL, LLS, ZL}.
Compound random measures (CoRM) \citep{GrLei16} are a unifying framework for many dependent random measures including many of the superposition and L\'evy copula approaches.  They have been applied to modelling graphs for overlapping communities by \cite{caron16}.
\cite{GrLei16} described posterior sampling methods for a particular class of normalized compound random measure mixtures which exploits a representation of the Laplace transform of a CoRM through a univariate integral of a moment generating function. 
 \cite{RanBlei15} independently developed a normalized CoRM model where the weights depend on a Gaussian process and described a variational Bayesian algorithm for inference.

In this paper, we will consider extending the class of compound random measures (CoRM)  from finite collections of distributions to infinite collections of distributions. 
This allows us to define CoRM models where the weights follow a time series model, the weights follow a regression model or the weights are defined through a hierarchical model. 
The computational algorithms in \cite{GrLei16}  cannot be used in this wider class of models since moment generating functions are not available in closed form.
Therefore, we develop a new MCMC algorithm for CRM-based nonparametric mixture models which uses a novel pseudo-marginal MCMC method \citep{AR2009}.

The paper is organized in the following way. Section 2 discusses defining NCoRM mixture models for distributions indexed by continuous covariates. Section 3 introduces a novel computational algorithm for NCoRM mixtures which can be widely applied. Section 4 illustrates how NCoRM can be used in density regression problems and how the computational algorithm performs. Section 5 concludes. Matlab code for the examples in this paper is available from \\ \verb+http://www.kent.ac.uk/smsas/personal/jeg28/index.htm+.

\section{Modelling with normalized compound random measure mixtures}

For simplicity, we will consider mixture models of the form in (\ref{DDP}) with 
$\theta_{k}(x)=\theta_{k}$ for all $x\in\X$, leading to a mixture model with weights which  vary  over $\X$ (many of the ideas in this paper could be extended to the model where $\theta_k(x)$ follows a stochastic process, such as a Gaussian process, over $\X$). The model is 
\begin{equation}
y_i\vert x_i \sim q(y_i \vert \theta_{c_i}),
\qquad p(c_i=k) = w_k(x_i),\quad k=1,\dots,\infty.
\label{mod1}
\end{equation}
We consider the weights
\begin{equation}\label{mod1W}
w_k(x) = \frac{m_{k}(x) J_k}{\sum_{l=1}^{\infty} m_l(x) J_l}
\end{equation}
where $m_k(x)$ is a random function on $\X$ for which $m_k(x)\geq 0$ for all $x\in\X$ 
and the function $m_k$ is independent of $m_l$, $J_1, J_2, J_3\dots,$ are the jumps of the process with directing L\'evy process $\nu^{\star}$ and $\theta_1, \theta_2, \theta_3, \dots$ are i.i.d.  We will refer to $m_k(x)$ as a {\it score} or {\it score function} and to $\nu^{\star}$ as the {\it directing L\'evy process}.
The model reduces to the NCoRM models considered by \cite{GrLei16} if $\X$ is a finite set. In particular, they introduced a  class of dependent random probability measures
$\tilde{p}_1,\tilde{p}_2,\dots,\tilde{p}_d$ which can be represented as 
\[
\tilde{p}_i
=\sum_{k=1}^{\infty} w_{ik}\delta_{\theta_k}
\]
with 
$$w_{ik}=\frac{m_{ik} J_k}{\sum_{k=1}^{\infty} m_{ik} J_k},$$
where $(m_{1k},\dots,m_{dk})$ are i.i.d. draws from a $d$-variate \textit{score} distribution $h$ (for $k=1,2,\dots$), $J_1, J_2, J_3\dots,$ are the jumps of the process with \textit{directing L\'evy process} $\nu^{\star}$ and $\theta_k\stackrel{i.i.d.}{\sim} \tilde{\alpha}$, with $\tilde{\alpha}=\alpha/\alpha(\mathbb{Y})$ where $\alpha$ is a positive finite measure. Under suitable conditions, the vector $(\tilde{p}_1,\tilde{p}_2,\dots,\tilde{p}_d)$ can be seen as a vector of normalized completely random measures, i.e. 
$$\tilde{p}_j=\frac{\tilde{\mu}_j}{\tilde{\mu}_j(\mathbb{Y})}\qquad j=1,...,d$$
where $(\tilde{\mu}_1,\tilde{\mu}_2,\dots,\tilde{\mu}_d)$ is a \textit{Compound Random Measure} (CoRM).
The model introduced in \eqref{mod1} assumes that $\mathbb{X}$ could potentially be a countable set. In this case, we assume that, for every finite subset $S=\{s_1,\dots,s_l\}$ of covariates, the $w(s_1),\dots,w(s_l)$ displayed in equation \eqref{mod1W} are the weights of a $l$-dimensional NCoRM process.

The specification of the weights displayed in equation \eqref{mod1W} has several attractive features which motivate our choice. Firstly, 
 the nonparametric approach allows the definition of a flexible model for density regression.
Secondly, the dependence between $w_k(x)$ and $w_k(x')$ for $x, x'\in\X$ can be controlled by the choice of the distribution for the random function $m_k$. Many methods have been developed to model such random functions and can be used to define a suitable dependent nonparametric mixture model.
Thirdly, the weights are not {\it a priori} stochastically ordered (as with many stick-breaking processes). Lastly, the structure of the model allows simpler computational methods to be developed than many other dependent extensions of normalized random measures.

We will concentrate on models where $m_l(x) = \exp\{r_l(x)\}$ and $r_l(x)$ is a random function on $\X$ taking value on $\mathbb{R}$. \cite{GrLei16} considered using the variance of the ratio of the same jump at values $x, x' \in \X$ as 
a simple measure of the strength of dependence between the (unnormalized) random measure at values $x$ and $x'$. In this case, the ratio is  $\zeta(x, x')=m_l(x)/m_l(x') = \exp\{r_l(x) - r_l(x')\}$ and the distribution of $\zeta(x, x')$ will often be easy to work with. For example,  $\zeta(x, x')$ will be log normally distributed if $r_l(x)$ and $r_l(x')$ have a bivariate normal marginal distribution. 

In this paper, we will consider models in which $r_l(x)$ is a stochastic process for which 
$\E[r_l(x)]=0$ for all $x\in\X$. This gives CoRM models a high degree of flexibility. 
To illustrate the use of NCoRM mixtures in a regression context, we will consider a choice of $r_l(x)$ which is suitable for continuous regressors and a choice of $r_l(x)$ which is suitable for categorical regressors:
\begin{itemize}
\item \textbf{Continuous regressors:} In this case, we define
$r_1(x), r_2(x), r_3(x), \dots$ to be independent Gaussian processes with covariance function $\sigma^2_0 \kappa(\cdot, \cdot)$ where $\kappa(\cdot,\cdot)$ is a correlation function.
This implies that $\log\zeta(\cdot, \cdot)$ follows a normal distribution with mean zero and variance $2\sigma_0^2$.
\item \textbf{Categorical regressors:} Suppose that we have two categorical regressors then
we could assume a different parameter for each combination of levels so that
$r_l(x_i) = \gamma^{(l)}_{x_{i,1},x_{i,2}}$. Alternatively, we could use the
 specification 
$r_l(x_i)=\alpha^{(l)}_{x_{i,1}}+\beta^{(l)}_{x_{i,2}}+\gamma^{(l)}_{x_{i,1},x_{i,2}}$ where 
\[
\alpha^{(k)}_j\sim\N(0,\sigma^2_1),
\qquad 
\beta^{(k)}_j\sim\N(0,\sigma^2_2),\qquad
\gamma^{(k)}_{i,j}\sim\N(0,\sigma^2_{1,2}).
\]
Then, the $\alpha^{(l)}$ and $\beta^{(l)}$ parameters act as main effects and $\gamma^{(l)}$ as interactions which can be interpreted in a similar way to a logistic regression model. For example, $\log\zeta(x,x')$ is normally distributed with mean $0$ and variance $2(\sigma_1^2+\sigma_2^2+\sigma_{1,2}^2)$ if both levels of $x$ are different to the levels of $x'$. Whereas, $\log\zeta(x,x')$ is normally distributed with mean $0$ and variance $2(\sigma_2^2+\sigma_{1,2}^2)$ if only the second level of $x$ and $x'$ are different. This shows how the dependence of jump sizes depends on the levels of the regressors. 
\end{itemize}
Posterior inference is impossible using existing methods and the following section describes a general purpose algorithm for NCoRM mixture models.

\section{Computational methods}

Posterior inference for nonparametric mixture models is challenging due to the infinite-dimensional random probability measure in the model. To address this problem, two main MCMC approaches to defining a finite-dimensional target have been developed.
 Firstly, marginal methods integrate the  random probability measure from the posterior. Secondly, conditional methods truncate the random probability measure. These methods can be further divided into exact methods which use a random truncation to sample exactly from the posterior and methods which fix the level of truncation leading to some truncation error. \cite{GrLei16} suggest a marginal method and an exact conditional method (a slice sampler). 
 The availability of an analytical expression for the moment generating function for the score distribution is key to their sampling methods but this is impossible to evaluate in closed form for the more general NCoRM models described in this paper. We propose a hybrid conditional-marginal sampler using a pseudo-marginal Metropolis-Hastings algorithm \citep{AR2009}.

We assume that we observe data $(x_1,y_1),\dots,(x_n,y_n)$ and wish to fit the model in (\ref{mod1}). Without loss of generality, we also assume that the values $x_1,x_2,\dots,x_n$ are distinct and write $m_{k,i}=m_k(x_i)$
 and $m_k=(m_{k,1},\dots,m_{k,n})$.
Following  \cite{GrLei16}, it is convenient to use an augmented
 form of the likelihood which introduces an allocation variable for each observation.
Let $n_k$ be the number of observations allocated to the $k$-th jump, we 
 order the jumps so that $J_1,\dots,J_K$ have points allocated to them ({\it i.e.} $n_k>0$ for $1\leq k\leq K$) and $J_{K+1},J_{K+2},\dots$ have no points allocated to them ({\it i.e.} $n_k=0$ for $k>K$).
Marginalizing over jumps which have no points allocated and the location of all atoms and writing $M=\alpha(\mathbb{Y})$ and $\tilde\alpha = \alpha/\alpha(\mathbb{Y})$ gives
\begin{align}\label{MargPost}
M^K\prod_{k=1}^K \left\{ J_k^{n_k}\left[\prod_{i=1; s_i=k}^n m_{k,i}\right]\exp\left\{-\sum_{i=1}^n v_i J_k \,m_{k,i}\right\}
\,h(m_k)\,\nu^{\star}(J_k)
\right\}
L
\prod_{k=1}^K q\left(y^{(k)}\right)
\end{align}
where
\begin{align}
L&=\E\left[\exp\left\{-\sum_{i=1}^n v_i \sum_{k=1}^{\infty} J_k\, m_{k,i}\right\}\right]\nonumber\\
&=
\exp\left\{-\int_{(\mathbb{R}^+)^n}\int_0^{\infty} \left(1-\exp\left\{-z\sum_{i=1}^n v_i\,m_i\right\}\right) h(m^{\star})\,\nu^{\star}(z)\,dz\,dm^{\star}\right\},
\label{laplace1}
\end{align}
\[
q(y)=\int \prod q(y_i\vert \theta)\,\tilde\alpha(\theta)\,d\theta.
\]
and $y^{(k)}=\{y\vert s_i=k,\ 1\leq i\leq n\}$. The expression in \eqref{MargPost} is the (marginalized) likelihood of the data. 
\cite{GrLei16} use the analytical expression for $L$ and integrals over $J_1,\dots,J_K$ to define a marginal sampler. In general, these integrals are not analytically available to us. We replace $L$ by an 
 unbiased estimate $\hat{L}$ (a possible unbiased estimator is discussed in the next Section) to define the following target
\[
M^K\prod_{k=1}^K \left\{ J_k^{n_k}\left[\prod_{i=1; s_i=k}^n m_{k,i}\right]\exp\left\{-\sum_{i=1}^n v_i J_k \,m_{k,i}\right\}
 \,h(m_k)\,\nu^{\star}(J_k)\right\}\hat{L}
\prod_{k=1}^K q\left(y^{(k)}\right).
\]
Finally, we assume that $h$ has parameters $\tau$ and $\nu^{\star}$ has parameters $\xi$ on which we want to make inference and define the target
\begin{align*}
&p(\tau)p(\xi)p(M)M^K\prod_{k=1}^K \left\{ J_k^{n_k}\left[\prod_{i=1; s_i=k}^n m_{k,i}\right]\exp\left\{-\sum_{i=1}^n v_i J_k \,m_{k,i}\right\}
 \,h(m_k\vert \tau)\,\nu^{\star}_{\xi}(J_k)\right\}\\
 &\times\hat{L}
\prod_{k=1}^K q\left(y^{(k)}\right).
\end{align*}

We propose a novel sampling strategy for the variable $s$ in a nonparametric mixture model and a novel computational algorithm to deal with the Laplace transform component of the target above. This algorithm can be applied to posterior inference for a wide variety of Bayesian nonparametric processes beyond NCoRM processes. 

\subsubsection*{Updating $c$}

To update $c_i$, we write the full conditional distribution as proportional to
\begin{align*}
&\prod_{k=1}^{K^-_i} \left(J_k\, m_{k,i}\, k(\{y_j\vert c_j=k\}\cup y_i)\right)^{\I(c_i=k)}\left(M\,m_{K^-_i+1,i}\,
\gamma\left(m_{K^-_i+1}\right)\, k(y_i)\right)^{\I(c_i=K^-_i+1)}\\
&\times p\left(J_{K^-_i+1}\left\vert m_{K^-_i+1}\right.\right)
h\left(\left. m_{K^-_i+1}\right\vert \tau\right)
\end{align*}
where $K^-_i$ is the number of distinct values in $c_1,\dots,c_{i-1},c_{i+1},\dots,c_n$, $(J_1,m_1),\dots,(J_K,m_K)$ are ordered so that $c_i=K$,
\[
\gamma(m_k)=\int z\,\exp\left\{-z\sum_{i=1}^n v_i\, m_{k,i}\right\}\nu_{\xi}^{\star}(z)\,dz
\]
and 
\[
p\left(J_{K^-_i+1}\left\vert m_{K^-_i+1}\right.\right)=\frac{J_{K^-_i+1}\exp\left\{-J_{K^-_i+1}\sum_{i=1}^n v_i\, m_{K^-_i+1, i}\right\}\nu_{\xi}^{\star}(J_{K^-_i+1})}{\gamma(m_{K^-_i+1})}.
\]
A new value of $J_{K_i^-+1}$ is sampled from this full conditional distribution leading to an algorithm which is similar to Algorithm 8 of \cite{Neal2000}. See \cite{JLP09}, \cite{LP_chapter} and \cite{Favaro2013} for extension to non-conjugate normalized random measure mixtures.

If the $i$-th observation was allocated to a singleton cluster in the previous iteration, the full conditional distribution of $c_i$ is 
\[
p(c_i=j) \propto \left\{
\begin{array}{ll}
J_j\, m_{j,i}\,
\frac{q(\{y_k\vert c_k=j\}\cup y_i)}{q(\{y_k \vert c_k=j\})}  &  j=1,\dots,K^-_i\\
M\,m_{j,i}\,
\gamma(m_{j})
\,q(\{y_i\})
& j = K^-_i+1
\end{array}\right.
\]

If the $i$-th observation was not allocated to a singleton cluster in the previous iteration, 
we propose $m_{K^-_i+1}\sim h\left(\left.m_{K^-_i+1}\right\vert \tau\right)$ and $J_{K^-_i+1}\sim p\left(J_{K^-_i+1}\left\vert m_{K^-_i+1}\right.\right)$, then
\[
p(c_i=j) \propto \left\{
\begin{array}{ll}
J_j\, m_{j,i}\,
\frac{q(\{y_k\vert c_k=j\}\cup y_i)}{q(\{y_k \vert c_k=j\})}  & j=1,\dots,K^-_i\\
M\,m_{K^-_i+1,i}\,
\gamma\left(m_{K^-_i+1}\right)\,q(\{y_i\})
& j = K^-_i+1
\end{array}\right..
\]
In Appendix \ref{CompDet} we provide the details of the full conditional distributions for the variables $J_k$, $m_k$, $v_i$, $\xi$, $M$ and $\tau$. The next Section will introduce the novel pseudo-marginal Metropolis-Hastings algorithm used to address the intractability of the Laplace transform part of the target distribution.

\subsection{Unbiased estimation of the Laplace functional}

\cite{AR2009} introduced a sampling scheme, called \textit{pseudo-marginal} Metropolis-Hastings, which allows sampling from distributions which cannot be evaluated pointwise. The main idea of the method is to replace the target distribution with a nonnegative unbiased estimator. 

In our framework,  we are often interested in evaluating objects such as the expectation  in (\ref{laplace1}),
\begin{align*}
L=\exp\left\{-\int_{(\mathbb{R}^+)^p}\int_0^{\infty} \left(1-\exp\left\{-z\sum_{i=1}^n v_i \,m(x_i)\right\}\right)p(m)\,\nu^{\star}(z)\,dz\,dm\right\}.
\end{align*}
We will use the Poisson estimator \citep{papa09} which has been successfully used in MCMC approaches for diffusions \citep[see {\it e.g.}][]{Fearnhead10}. Consider, the equation
\begin{equation}
L_{\phi} =
 \exp\left\{-\int_D \phi(x) \,dx\right\}\leq 1
\label{main_eq}
\end{equation}
for $D\subset \mathbb{R}^p$
where $\phi(x)\geq 0$ for all $x\in D$ and $\int_D \phi(x)\,dx<\infty$. 
\medskip
The Poisson estimator of \eqref{main_eq} is introduced in the following Theorem where some properties are described. The proof of the Theorem can be found in the Appendix. We denote  the Poisson distribution with parameter $\lambda$ by $Pn(\lambda)$. 

\begin{thm}\label{PoissonThm}
Consider the following estimator, 
  
\begin{equation}
\hat{L}_{\phi} = \prod_{i=1}^K \left(1 - \frac{\phi(x_i)}{a\,C\,\kappa(x_i)}\right)
\label{est_eq}
\end{equation}
where $\kappa$ is a p.d.f. on $D$, $C>\frac{\phi(x)}{\kappa(x)}$
for $x\in D$, $a>1$, $K\sim\Pn(a\,C)$ and $x_i\stackrel{i.i.d.}{\sim} \kappa$. Then, 

$$\E[\hat{L}_{\phi}]=\exp\left\{-\int_D \phi(x)\,dx\right\}$$
and
\begin{align*}
\V[\hat{L}_{\phi}]&=
L_{\phi}^2
\left(\exp\left\{
\frac{1}{a\,C}\int_D\frac{\phi(x)^2}{\kappa(x)}\,dx\right\}-1\right)\\
&\leq 
L_{\phi}^2
\left(\exp\left\{
\frac{1}{a}\int_D \phi(x)\,dx\right\}-1\right)
<\infty.
\end{align*}
\end{thm}


The estimator has the useful property that it is always positive. This contrasts with other approaches which define unbiased estimators of infinite sums using random truncation where it is difficult to ensure that estimates are always positive \citep[see {\it e.g.}][]{rhee15, Lyne15}.

Returning to the expression in  (\ref{laplace1}) and, again, assuming that $x_1,x_2,\dots,x_n$ are distinct, this can be re-expressed as
\begin{align}
L=& \exp\left\{-\int_{(\mathbb{R}^+)^n}\int_0^{\infty} \left(1-\exp\left\{-z\sum_{i=1}^n v_i \,m^{\star}_i\right\}\right)h(m^{\star})\,\nu^{\star}(z)\,dz\,dm^{\star}\right\}\nonumber\\
=& 
 \exp\left\{-\int_{(\mathbb{R}^+)^n}\int_0^{\infty} \sum_{k=1}^n v_k\, m^{\star}_k\int_0^z \exp\left\{-t\sum_{i=1}^n v_i\,m^{\star}_i\right\}\,dt\, h(m^{\star})\,\nu^{\star}(z)\,dz\,dm^{\star}\right\}
\nonumber\\
=&
 \exp\left\{-\int_{(\mathbb{R}^+)^n}\int_0^{\infty} \sum_{k=1}^n v_k\,m^{\star}_k\, h(m^{\star})\exp\left\{-t\sum_{i=1}^n v_i\,m^{\star}_i\right\}\,\int_t^{\infty} \nu^{\star}(z)\,dz\,dt\,dm^{\star}\right\}\nonumber\\
 =&
 \exp\left\{-\int_{(\mathbb{R}^+)^n}\int_0^{\infty} \sum_{k=1}^n v_k\, m^{\star}_k\, h(m^{\star})\exp\left\{-t\sum_{i=1}^n v_i\,m^{\star}_i\right\}\,T_{\nu^{\star}}(t)\,dt\,dm^{\star}\right\}\nonumber\\
 =&\prod_{k=1}^n L_k.
\label{NCoRM1}
\end{align}
where $T_{\nu^{\star}}(t)=\int_t^{\infty} \nu^{\star}(z)\,dz$ is the tail mass function for the L\'evy process with L\'evy intensity $\nu^{\star}$ and
\[
L_k=
 \exp\left\{-\int_{(\mathbb{R}^+)^n}\int_0^{\infty} v_k\, m^{\star}_k\,h(m^{\star})\exp\left\{-t\sum_{i=1}^n v_i\,m^{\star}_i\right\}\,T_{\nu^{\star}}(t)\,dt\,dm^{\star}\right\}
.
\]
The expression for $L_k$ has the form of (\ref{main_eq}) with $x=(z,m^{\star}_k)$, $D=(0,\infty)\times {\mathbb{R}^+}^n$ and \[
\phi(z,m^{\star}_k)= v_k\, m^{\star}_k\, h(m^{\star}_k)\exp\left\{-z\sum_{i=1}^n v_i\,m^{\star}_k\right\}\,T_{\nu^{\star}}(z).
\] 
Clearly $\int_D \phi(x)\,dx<\infty$. To use the Poisson estimator, a suitable density is 
\begin{equation}
\kappa(z,m^{\star}_k)=
\kappa_{\nu^{\star}}(z)\frac{m^{\star}_k\,h(m^{\star}_k)}
{\E[m^{\star}_k]}
\label{prop_levy}
\end{equation}
where $\kappa_{\nu^{\star}}(z)>T_{\nu^{\star}}(z)$ for all $z\in \mathbb{R}^+$. We use
$
C = v_k\,\E[m_k^{\star}]\,B
$
where $\frac{T_{\nu^{\star}}(z)}{\kappa_{\nu^{\star}}(z)}<B$ for all $z$.
 Suitable forms of $\kappa_{\nu^{\star}}$ for some popular nonparametric processes are given in Section \ref{ExamplesPoissonEst}.

In computation for more usual normalized random measures \citep{GW2011, Favaro2013}, 
we are interested in 
\begin{equation}\label{Peq2}
\E[\exp\{-vJ\}] = \exp\left\{-\int_0^{\infty} (1-\exp\{-vz\})\nu(z)\,dz\right\}\leq 1
\end{equation}
where $\nu(z)$ is a L\'evy process and the expectation is taken over all jumps on  $\mathbb{R}^+$. This expectation can, similarly, be  re-expressed as
\begin{align*}
= \exp\left\{-\int_0^{\infty} (1-\exp\{-vz\})\nu(z)\,dz\right\}
&=\exp\left\{- v\int_0^{\infty} T_{\nu}(t)\,\exp\{-vt\}\,dt\right\}
\end{align*}
which is (\ref{main_eq}) with $D=\mathbb{R}^+$ and $\phi(x)=v T_{\nu}(x)\exp\{-vx\}$. The estimator in (\ref{est_eq})
provides an unbiased estimator of \eqref{NCoRM1} and \eqref{Peq2} which can be used in the sampler described in this section.

\subsubsection{Controlling the variability of $\hat{L}_{\phi}$}

Pseudo-marginal Metropolis-Hastings algorithms converge to the correct distribution but the asymptotic variance of an average calculated using the algorithm depends on the variance of the unbiased estimator. For example, suppose that the unbiased estimator is an importance sampler. \cite{AV2016} show that the asymptotic variance of the pseudo-marginal sampler  decreases as the number of samples in the importance sampler increases (leading to an importance sampler with a lower asymptotic variance). Although we do not use an importance sampler, the Poisson estimator is closely related and it is intuitively reasonable that the asymptotic variance of averages calculated using the pseudo-marginal Metropolis-Hastings sampler will decrease as the variance of the Poisson estimator in \eqref{est_eq} decreases.

The variability of the Poisson estimator is controlled by $a$ with larger values of $a$ leading to a smaller variance. However, larger values of $a$ will also lead to longer computational times since the mean number of terms in $\hat{L}_{\phi}$ is $aC$. In this section, we will assume that the expected number of evaluations of the ratio $\phi(x)/\kappa(x)$  is $d$. Therefore, $d=aC$ for the estimator in 
\eqref{est_eq}. An alternative method for controlling the variability involves defining
the estimator $\hat{L}^{AVE}_{\phi}=\frac{1}{N}\sum_{i=1}^N \hat{L}^{(i)}_{\phi}$ where $\hat{L}^{(1)}_{\phi},\dots,\hat{L}^{(N)}_{\phi}$ are independent realisations of $\hat{L}_{\phi}$. The estimator has variance
\[
\V[\hat{L}_{\phi}^{AVE}]=
\frac{L_{\phi}^2}{N}
\left(
\exp\left\{\frac{1}{a\,C}\int_D\frac{\phi(x)^2}{\kappa(x)}\,dx\right\}-1\right)
=\frac{L_{\phi}^2}{N}
\left(
\exp\left\{\frac{N}{d}\int_D \frac{\phi(x)^2}{\kappa(x)}\,dx\right\}-1\right).
\]
since $d=a\,C\,N$.
It is straightforward to show that 
\[
\frac{L_{\phi}^2}{N+1}
\left(
\exp\left\{\frac{N+1}{d}\int_D \frac{\phi(x)^2}{\kappa(x)}\,dx\right\}-1\right)
>
\frac{L_{\phi}^2}{N}
\left(
\exp\left\{\frac{N}{d}\int_D \frac{\phi(x)^2}{\kappa(x)}\,dx\right\}-1\right)
\]
and so the variance of the estimator grows with $N$ for fixed $d$. This suggests that we should use 
$\hat{L}^{AVE}_{\phi}$ with $N=1$ which is the Poisson estimator in \eqref{est_eq}.
In this case, the choice of $\kappa(x)$ which minimizes the variance for fixed $d$ is $\kappa(x) = \frac{\phi(x)}{-\log L_{\phi}}$ which provides a criterion for choosing $\kappa(x)$.

As we have already mentioned the asymptotic variance of averages calculated using the pseudo-marginal algorithm will typically decreases as $a$ increases but the computational time will increase. Therefore, there is an optimal value of $a$ which is able to provide the lowest asymptotic variance for a fixed computational budget (number of evalulations of $\phi(x)/\kappa(x)$). \cite{Doucet} established an upper bound for the asymptotic variance under certain assumptions which allows this optimal value to be derived. They demonstrated that this value can be close to optimal when the assumptions are violated. In our context, their main assumption is
\[
\hat{L}_{\phi} = \epsilon L_{\phi} 
\]
where $\log\epsilon\sim\N(-\sigma^2/2,\sigma^2)$. They refer to $\sigma^2$ as the noise variance and it is straightforward to show that 
\[
\sigma^2
=\frac{1}{a\,C}\int_D  \frac{\phi(x)^2}{\kappa(x)}\,dx.
\]
\cite{Doucet} showed that the optimal value of the noise variance (in terms of asymptotic variance), $\sigma^2_{opt}$, depends on the properties of the chain but provide guidelines on how this can be approximated.  Following the derivation of \cite{Doucet}, the optimal value of $a$, for fixed $\kappa$, is
\[
a_{opt}=\frac{1}{\sigma^2_{opt}\,C}\int_D  \frac{\phi(x)^2}{\kappa(x)}\,dx.
\]
In practice, we have found that the value $a=8$ works well for the processes considered in this paper.

\subsubsection{Examples}\label{ExamplesPoissonEst}

\cite{Brix99} provided a bound for the tail-mass integral of the generalized gamma process which is extended to the stable-Beta process by \cite{AP2016}. However, both bounds are not tight and we suggest tighter bounds for both processes. Indeed, the estimator introduced in Theorem \ref{PoissonThm} with the proposal in
\eqref{prop_levy}
requires draws from a Poisson distribution whose mean is proportional to $v_k\,\E[m_k^{\star}]\,D$
where $\frac{T_{\nu^{\star}}(z)}{\kappa_{\nu^{\star}}(z)}<B$ for all $z$. Therefore, better choices of $\kappa_{\nu^{\star}}(x)$ can improve the computational efficiency of the method by requiring a smaller value of $B$.

\bigskip



%

\noindent\textit{Generalized gamma process}

\bigskip
The generalized gamma process has L\'evy density $\nu^{\star}(y)=\frac{1}{\Gamma(1-\sigma)}y^{-1-\sigma}\exp\{-\lambda y\}$ and the tail-mass functions $T_{\nu^{\star}}(t)$ is an incomplete gamma function. It is straightforward to show that $T_{\nu^{\star}}(t) < \tilde\kappa_{\nu^{\star}}(t)$ where
\[
\tilde\kappa_{\nu^{\star}}(t) = \left\{
\begin{array}{ll}
\frac{1}{\sigma\Gamma(1-\sigma)} (t^{-\sigma}-1) & t<b\\
\frac{1}{\sigma\Gamma(1-\sigma)} (b^{-\sigma}-1)\exp\{-\lambda(t - b)\} & t\geq b
\end{array}\right.
\]
and $b=0.65$. A suitable choice of $D$ is $D=\frac{1}{\sigma\Gamma(1-\sigma)}\left[\frac{b^{1-\sigma}}{1 - \sigma} - b + b^{-\sigma} - 1\right]$ and the
bounding p.d.f. is $\kappa_{\nu^{\star}}(t)=\frac{1}{D}\tilde\kappa_{\nu^{\star}}(t)$.
  Consider $\kappa_{\nu^{\star}}(t)$ truncated to $t<b$, taking the transformation $y=\frac{1}{\sigma}(t^{-\sigma}-1)$  leads to the density
 $\kappa(y)\propto y(\sigma y+1)^{-1/\sigma-1}$ truncated to $y>\frac{1}{\sigma}(b^{-\sigma}-1)$. This can be expressed as a mixture of gamma distributions where $y\vert \Xi\sim\Ga(2, \Xi)$, $\Xi\sim\Ga(1/\sigma - 1,1/\sigma)$. 
 
As $\sigma\rightarrow 0$ for $\lambda=1$,  the generalized gamma process converges to the gamma process which has L\'evy density $\nu(z)=z^{-1}\,\exp\{-z\}$ and
$
T_{\nu}(t)=\int_t^{\infty} \nu(y)\,dy = \mbox{E}_1(t)
$
where $\mbox{E}_1(t)$ is the exponential-integral function. Both the bounding p.d.f. and simulation scheme for $t<b$ also converge. It is straightforward to show that the limit is
\[
\tilde\kappa_{\nu^{\star}}(t) = \left\{
\begin{array}{ll}
-\log t & t<b\\
-(\log b) \exp\{-(t - b)\} & t\geq b
\end{array}\right..
\]
with $D = b - b \log b - \log b$. The appropriate transformation for $t<b$ is  $y = -\log(t)$ which has p.d.f. $z\exp\{-z\}$, {\it i.e.} $z\sim\Ga(2,1)$ truncated to $y> -\log b$.

If $\lambda=0$, the generalized gamma process is stable process. However, the tail mass function is infinite for a stable process and this simulation scheme is not possible.

\bigskip

\noindent\textit{Stable-Beta process}

\bigskip

The stable-Beta process has L\'evy density 
$\nu^{\star}(z)=\frac{\Gamma(\phi)}{\Gamma(\sigma+\phi)\Gamma(1-\sigma)} z^{-\sigma-1}(1-\lambda\,z)^{\sigma+\phi-1}$ for  $0<z<1/\lambda$.
It is straightforward to show that $T_{\nu^{\star}}(t) < \tilde\kappa_{\nu^{\star}}(t)$ where
\[
\tilde\kappa_{\nu^{\star}}(t) = \left\{
\begin{array}{ll}
\frac{\Gamma(\phi)}{\sigma\Gamma(\sigma+\phi)\Gamma(1-\sigma)} 
 (t^{-\sigma}-1) 
 & t<b\\
 \frac{\Gamma(\phi)}{\sigma\Gamma(\sigma+\phi)\Gamma(1-\sigma)} 
 (b^{-\sigma}-1) 
\frac{(1 -\lambda t)^{\gamma+\sigma}}{(1 - \lambda b)^{\gamma+\sigma}}
 & t\geq b
\end{array}\right.
\]
and $b=0.65$.
A suitable choice of $D$ is $D=
\frac{\Gamma(\phi)}{\Gamma(\sigma+\phi)\Gamma(1-\sigma)} \left[\frac{b^{1-\gamma}}{1-\gamma} - b + \frac{b^{-\sigma}-1}{\lambda(\gamma+\sigma+1)}(1-\lambda b)\right]$
 and the
bounding p.d.f. is $\kappa_{\nu^{\star}}(t)=\frac{1}{D}\tilde\kappa_{\nu^{\star}}(t)$. 
As $\sigma\rightarrow 0$ for $\lambda=1$, the stable-Beta process converges to the Beta process which has L\'evy density $\nu(y)=\frac{1}{\Gamma(\gamma)}y^{-1}(1 - y)^{\gamma-1}$ for $0<y<1$. In this case, the limit of $\tilde\kappa_{\nu^{\star}}(t)$ is
\[
\tilde\kappa_{\nu^{\star}}(t)  = \left\{
\begin{array}{ll}
-\log t & t<b\\
\frac{-\log b}{(1-b)^{\gamma}}(1-t)^{\gamma}
 & t\geq b
\end{array}\right.
\]
and the limit of $D$ is $D=b - b \log b - \frac{\log b}{\gamma+1}(1-b)$.

\section{Illustrations}

\subsection{Example 1: Discrete regressors}

The algorithms developed in this paper are illustrated using an analysis of hematological data arising from a dose-escalation study which has previously been analysed by \cite{MuRo97}. The data are white blood cell counts over time for a sample of 52 patients receiving different levels of two treatments: cyclophosphamide (CTX) and a second drug (GM-CSF). The data for each patient is summarized as the maximum likelihood estimates from a non-linear regression model with seven parameters  fitted to that patient's time profile. 
The  model assumes that the mean response at time $t$ with parameters $\theta=(z_1,z_2,z_3,\tau_1,\tau_2,\beta_1)$ is given by
\[
f(\theta,t)=\left\{
\begin{array}{ll}
z_1 & t<\tau_1\\
rz_1+ (1 - r) g(\theta,\tau_2) & \tau_1\leq t<\tau_2\\
g(\theta,t) & t\geq \tau_2
\end{array}
\right.
\]
where $r=(\tau_2-t)/(\tau_2-\tau_1)$ and $g(\theta, t)=z_2 + z_3
/[1+\exp\{2.0 - \beta_1(t-\tau_2)\}]$. The model implies that the white blood cell count is constant (at level $z_1$) before $\tau_1$ followed by a linear progression between $\tau_1$ and $\tau_2$ and a logistic recovery after $\tau_2$. The parameters $z_2$ and $z_3$ control the white blood cell count at the start and end of recovery.
\cite{deiorio} applied an ANOVA-DDP model to these data which assumes a mixture model with constant weights and an ANOVA model for the locations for each treatment. In contrast, we fitted a mixture model with weights that vary with the treatment combination but with locations that do not depend on the treatment level. Specifically, we assume that $y_i$ are the estimated parameters for the $i$-th patient and that $x_{i,1}$ is the level of CTX and $x_{i, 2}$ is the level of GM-CSF. The model is 
\[
y_i\vert s_i \sim\N(\mu_{s_i}, \Sigma),\qquad \mu_k\sim\N(\mu_0, \lambda^{-1}\Sigma),
\]
\[
\Sigma^{-1}\sim\mbox{W}(\nu, \Psi)
\]
\[
p(c_i=k)=\frac{J_k \exp\left\{\alpha^{(k)}_{x_{i,1}}+\beta^{(k)}_{x_{i,2}}+
\gamma^{(k)}_{x_{i,1},x_{i,2}}\right\}}{\sum_{l=1}^{\infty} J_l \exp\left\{\alpha^{(l)}_{x_{i,1}}+\beta^{(l)}_{x_{i,2}}+\gamma^{(l)}_{x_{i,1},x_{i,2}}\right\}}
\]
\[
\alpha^{(k)}_j\sim\N(0,\sigma^2_1),
\qquad 
\beta^{(k)}_j\sim\N(0,\sigma^2_2),\qquad
\gamma^{(k)}_{i,j}\sim\N(0,\sigma^2_{1,2}).
\]
The directing L\'evy process is taken to be a gamma process. The model assumes a two-way ANOVA model with interaction for the logarithm of the weights. This does not place restriction on the combination of weights but does encourage similar weights for similar combinations of levels. The priors were
$\sigma_1^{2}\sim\mbox{Ga}(1,2)$,
$\sigma_2^{2}\sim\mbox{Ga}(1,2)$,
$\sigma_{1,2}^{2}\sim\mbox{Ga}(1,2)$ and $M\sim\mbox{Ga}(1, 1)$. For the purposes of illustration, we set $\mu_0$ equal to the sample mean of the data,
$\Psi = \frac{1}{9(\nu - 8)}\hat\Sigma$ where $\hat\Sigma$ is the covariance of the data which implies that the prior mean  of  $\Sigma$ is $\frac{1}{9}\hat\Sigma$ and we choose $\lambda=0.01$.
 The MCMC algorithms was run for a total of 35\ 000 iterations. The first 5\ 000 were used as a burn-in with the subsequent values thinned every fifth sample. This gave a sample of 6\ 000 values.

\begin{figure}[h!]
\begin{center}
\includegraphics[clip, trim=0mm 00mm 80mm 210mm]{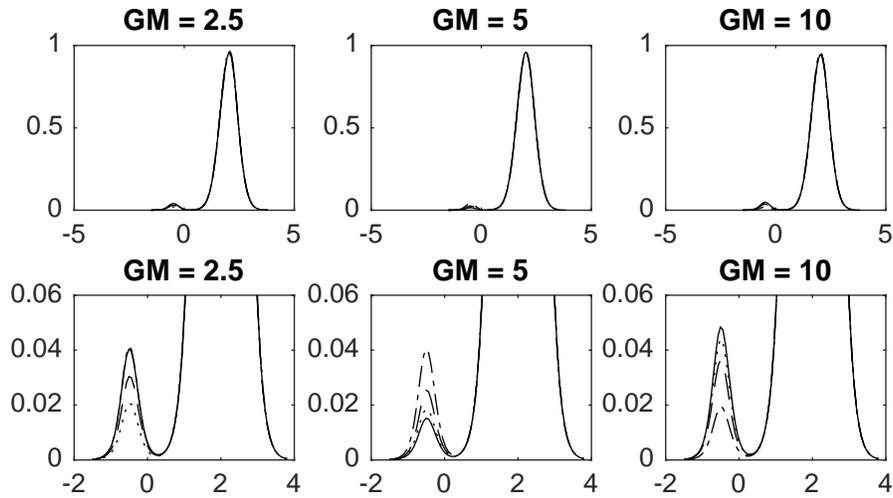}
\end{center}
\caption{Posterior mean probability density of $z_1$ for three levels of GM-CSF and four levels of CTX shown as: solid line (CTX=1.5), dashed line (CTX=3.0), dot-dashed line (CTX=4.5) and dashed line (CTX=6.0)}\label{group_var1}
\end{figure}

\begin{figure}[h!]
\begin{center}
\includegraphics[clip, trim=0mm 00mm 80mm 210mm]{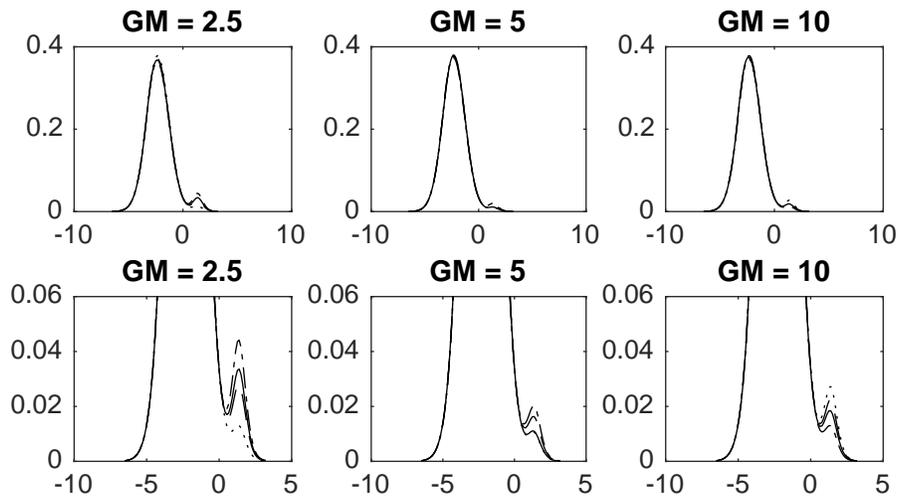}
\end{center}
\caption{Posterior mean probability density of $z_2$ for three levels of GM-CSF and four levels of CTX shown as: solid line (CTX=1.5), dashed line (CTX=3.0), dot-dashed line (CTX=4.5) and dashed line (CTX=6.0)}\label{group_var2}
\end{figure}
The inference about the marginal probability of two parameters $z_1$ and $z_2$  are shown in Figures \ref{group_var1} and \ref{group_var2}. The parameter $z_1$ is the initial white blood cell count. The distribution is bi-modal with the size of the smaller mode increasing with GM-CSF. This indicates that there are differences in the proportion of patients with lower white blood cell count across the different treatments.
 The parameter $z_2$ controls the level of white-blood cells when recovery begins and this is again bi-modal with the smaller mode decreasing  with GM-CSF for $z_2$.

\begin{figure}[h!]
\begin{center}
\includegraphics[trim=0mm 0mm 80mm 210mm, clip]{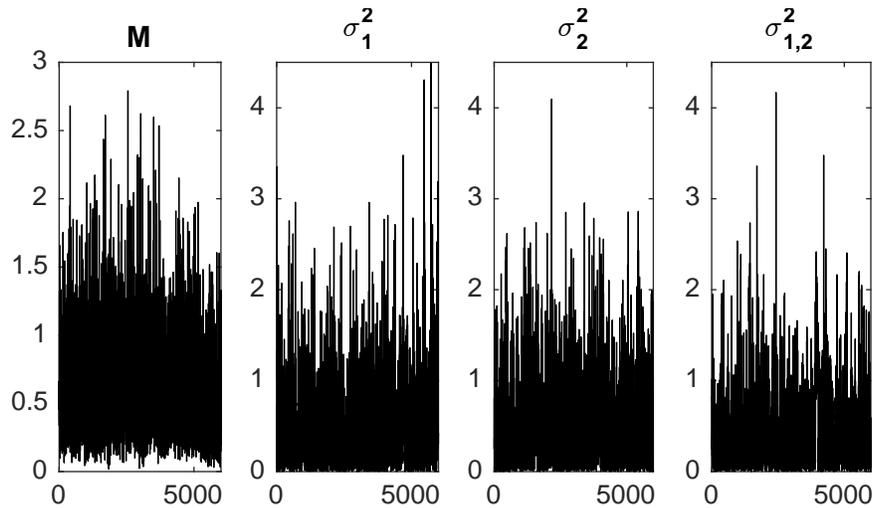}
\end{center}
\caption{Trace plots of the parameters $M$, $\sigma^2_1$, $\sigma^2_2$ and $\sigma^2_{1,2}$.}\label{group_trace}
\end{figure}
Figure~\ref{group_trace} shows trace plots for the total mass parameter $M$ and the three parameters controlling the differences between jumps at each treatment level. These clearly show good performance of the sampler for this problem.

\subsection{Example 2: Continuous Regressors}

A regression model is used to define an infinite mixture model with regressor dependent weights. We observe pairs $(x_1,y_1),\dots,(x_n,y_n)$ where $x_i\in\mathbb{R}^p$ and $y_i\in\mathbb{R}$ and use the model
\[
y_i\sim \N(\theta_{s_i}, a\sigma^2),\qquad \theta_k\sim\N(\mu, (1-a)\sigma^2),
\]
\[
p(c_i=k)=\frac{J_k\, \exp\{r_k(x_i)\}}
{\sum_{l=1}^{\infty} J_l\,\exp\{r_l(x_i)\}}
\]
where $r_1(x), r_2(x), r_3(x), \dots$ are independent Gaussian processes. A generalized gamma directing L\'evy process is used with $\lambda=1$ and three values of $\sigma$: $\sigma=0$ (a gamma process), $\sigma=0.1$ and $\sigma=0.5$.

We apply the model to data from a simulated motorcycle accident used to test crash helmets \citep{silverman85}, which are available as the \verb+mcycle+ data frame in the R package MASS. The data are head accelerations (in $g$)  measured at different times in milliseconds after impact.
We assume that the Gaussian processes have covariance function 
$C(x,y) = \phi\exp\left\{-\frac{\parallel x-y\parallel}{L}\right\}$ where $\parallel \cdot\parallel$ is Euclidean distance and $L$ is the lengthscale. The priors are $\alpha\sim\mbox{U}(0,1)$, $p(\mu,\sigma^2)\propto \sigma^{-2}$, $L\sim\mbox{Ga}(1,1)$, $M\sim\mbox{Ga}(1, 1)$, and $\phi^{-1}\sim\mbox{Ga}(1,4)$. The prior for $\phi$ is chosen so that $\log m_k(x_i)$ typically takes values in $(-4,4)$. The MCMC algorithms was a total of 33\ 000 iterations. The first 3\ 000 were used as a burn-in with the subsequent values thinned every third sample. This gave a sample of 10\ 000 values.

\begin{figure}[h!]
\begin{center}
\begin{tabular}{ccc}
\includegraphics[trim=0mm 0mm 160mm 260mm, clip]{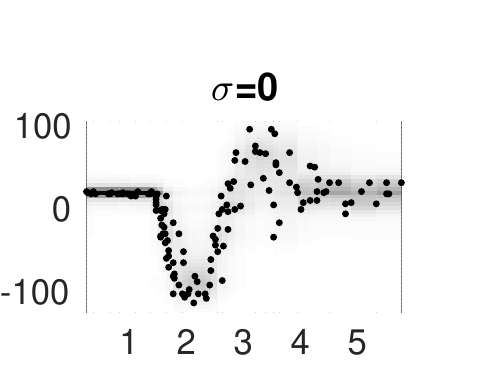}
\includegraphics[trim=0mm 0mm 160mm 260mm, clip]{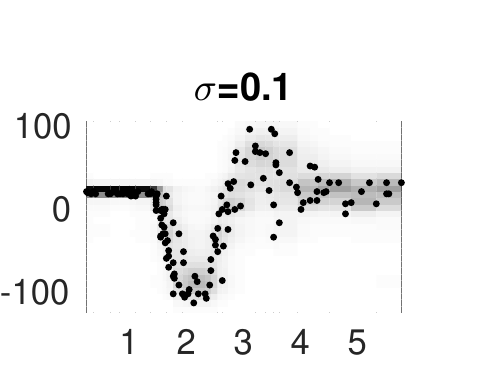}
\includegraphics[trim=0mm 0mm 160mm 260mm, clip]{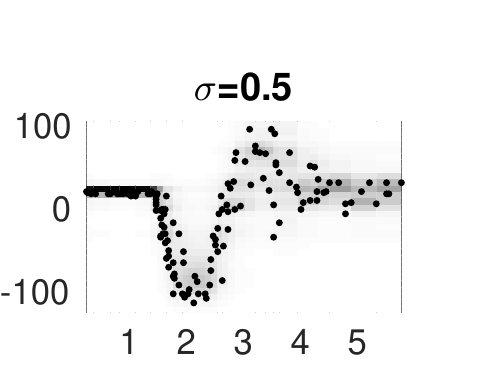}
\end{tabular}
\end{center}
\caption{Motorcyle data: Data (dots) and posterior mean density of $y|x$ (darker colours show larger density) with $\sigma=0$, $\sigma=0.1$ and $\sigma=0.5$.}\label{mcycle_pred}
\end{figure}
Figure~\ref{mcycle_pred} shows the posterior mean of the conditional density of head acceleration given time from impact for the three values of $\sigma$ with the data superimposed. In each case, the model was able to follow the data and capture the changing the heterogeneity in the variance. The inference seems robust to the choice of $\sigma$.

\begin{figure}[h!]
\begin{center}
\includegraphics[trim=10mm 0mm 50mm 260mm, clip]{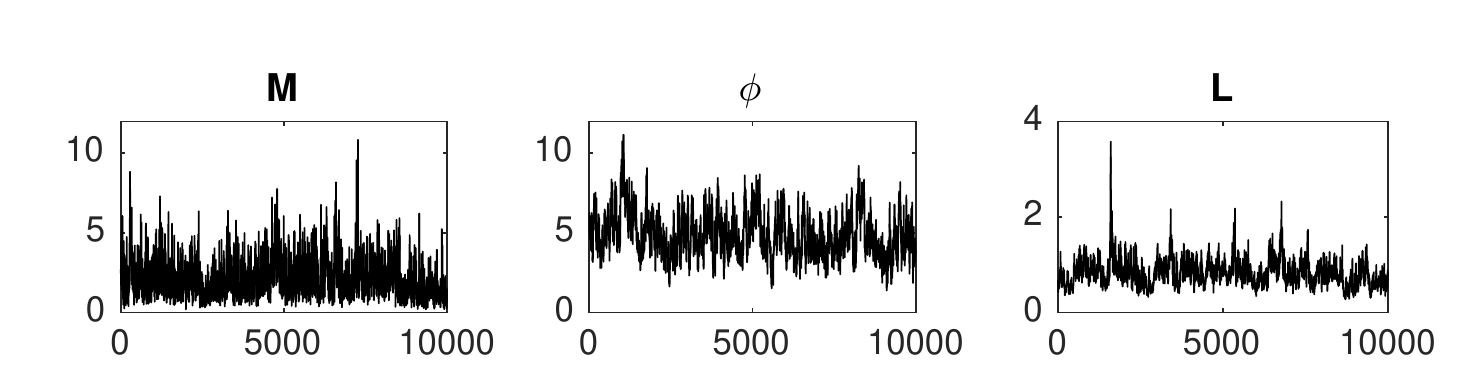}
\end{center}
\caption{Motorcyle data: Trace plots of the parameters $M$, $\phi$ and $L$ for $\sigma=0$ (the gamma process).}
\label{mcycle_trace}
\end{figure}
Trace plots for the three parameters $M$, $\phi$ and $L$ for the case $\sigma=0$ (the gamma process) are shown in figure~\ref{mcycle_trace}. These clearly show that the parameters are mixing well across the MCMC chain.

\subsection{Comparison of predictive performance}

We ran a simulation exercise to understand how the NCoRM regression model developed in this paper compared to two commonly used dependent nonparametric priors: the single-$p$ dependent Dirichlet process \citep{deiorio} and a probit stick-breaking process mixture \citep{roddun11}. The methods were compared by 10-fold cross-validation using simulated data sets of size 100 (leading to training data sets with 90 observations) according to the 
 out-of-sample log-predictive scores, {\it i.e.}
\[
\mbox{LPS}=-\frac{1}{100}\sum_{i=1}^{10} \sum_{j=1}^{10}
\log p\left(\left.y^{test}_{i,j}\right\vert y^{train}_{i,1},\dots,y^{train}_{i,90}\right).
\]
where $y^{test}_{i,1},\dots,y^{test}_{i,10}$ is the $i$-th testing sample and 
$y^{train}_{i,1},\dots,y^{pred}_{i,90}$ is the $i$-th training sample. Three sets of data were simulated to cover different modelling situations. The first two data sets used regressor-dependent mixture models covering the simple case of a two component mixture and a more complicated scenario with four components. The third dataset used a non-linear regression model which depends on four parameters $a$, $b$, $c$ and $d$. Different values of the parameters lead to different features of the data such as homoscedascity or heteroscedascity, jumps or different levels of smoothness. In all cases, there was a single regressor which was generated uniformly on $(0,1)$. The detailed descriptions of the data sets are given below.

\begin{itemize}

\item Simulated Data Sets I

The responses were simulated as
\[
y\sim
\left\{\begin{array}{ll}
\N(-1, \sigma^2) &\mbox{ if }s=1\\
\N(1, \sigma^2) & \mbox{ if }s=2
\end{array}\right.
\]
where
\[
p(s=1)=\frac{\exp\{r\sin(2\pi x)\}}{1+\exp\{r\sin(2\pi x)\}},\qquad
p(s=2)=\frac{1}{1+\exp\{r\sin(2\pi x)\}}.
\]
We consider two values of $\sigma$ (0.1 and 0.5) which allow for different levels of separation between the two mixture components and two values of $r$ (1 and 2) which control the rate at $p(s=1)$ changes over the range of $x$.

\item Simulated Data Sets II 

The responses were simulated as
\[
y
\sim\left\{\begin{array}{ll}
\N(-1, \sigma^2) & \mbox{ if }s=1\\
\N(1, \sigma^2) & \mbox{ if }s=2\\
\N(-2, \sigma^2) & \mbox{ if }s=3\\
\N(2, \sigma^2) &  \mbox{ if }s=4
\end{array}\right.
\]
where $
p(s=1)\propto \exp\{2\sin(2\pi x)\}$, $p(s=2)\propto 1/2 + 2/5 (x - 1/2)$,
$p(s=3)\propto 1/2 - 2 (x - 1/2)^2$ and $p(s=4)\propto 1$. 
Again, we consider two values of $\sigma$ (0.1 and 0.5)  to give different levels of separation between the clusters.

\item Simulated Data Sets III 

The responses were simulated as
\[
y = g_{a,b}(x) + \I(c=1) h(x) +
0.1\,\epsilon\,k(x)^{\I(d=1)}
\]
where 
\[
g_{a,b}(x) = \left\{
\begin{array}{cc}
0 & x<a\\
\sin(2\pi (x-a)/(b-a)) & a \leq x\leq b\\
0 & x>b
\end{array}
\right.,
\]
\[
h(x) = \left\{
\begin{array}{cc}
-2 & x < 1/3\\
-3/4 & 1/3 < x < 3/4\\
0 & x>3/4
\end{array}
\right.,
\]
$k(x) = 0.15(1 + 19|\sin (2\pi)|)$
and $\epsilon\sim\N(0,1)$. Different choices of $a$, $b$, $c$ and $d$ lead to responses which have a nonlinear mean and potentially heteroscedascity. If $c=0$ and $d=0$, the means of responses have a sine wave scaled to the interval $(a,b)$ (with zero mean outside the interval) with homoscedastic noise.  Additionally, the mean of the responses jumps at $1/3$ and $3/4$  if $c=1$ and the errors are heteroscedastic if $d=1$.

\end{itemize}

\begin{table}[h!]
\begin{center}
\begin{tabular}{|c|c|ccc|}\hline
    $\sigma$   &  $r$   & 	Probit SB       &	DDP &	NCoRM\\\hline
 0.1 &  2 &	-0.33	& -0.18      &	-0.36\\
 0.1  &  1 &	0.30	& 0.51       &	0.26\\
 0.5  & 2 &	1.35	& 1.41       &	1.26\\
 0.5  & 1 &	1.53	& 1.54       &	1.50\\\hline
\end{tabular}
\end{center}
\caption{LPS of the probit stick-breaking, dependent Dirichlet process and NCoRM mixtures for Simulated Data Sets I.}
\label{sim_data_I}
\end{table}
The LPS for the three different nonparametric priors with Simulated Data Sets I are given in Table~\ref{sim_data_I}. The NCoRM mixture outperformed both the DDP and probit stick-breaking processes for the four combination of $\sigma$ and $r$.
\begin{table}[h!]
\begin{center}
\begin{tabular}{|c|ccc|}\hline
    $\sigma$      & 	Probit SB       &	DDP &	NCoRM\\\hline
 0.1  &  -0.22   & -0.15      &	-0.22\\
 0.5  &	1.49	& 1.54       &	1.48\\\hline
\end{tabular}
\end{center}
\caption{LPS of the probit stick-breaking, dependent Dirichlet process and NCoRM mixtures for Simulated Data Sets II.}
\label{sim_data_II}
\end{table}
The LPS's for Simulated Data Sets II (Table~\ref{sim_data_II}) show that the DDP outperformed the other methods for $\sigma=0.1$ and the NCoRM outperformed the other methods for $\sigma=0.5$.
\begin{table}[h!]
\begin{center}
\begin{tabular}{|c|c|c|c|ccc|}\hline
$a$ & $b$   &   $c$   &   $d$   &     Probit SB	 &   DDP &	NCoRM\\\hline
0 & 1 &   0  &  0 &  0.80 &  -0.81  &  -0.31\\
0 & 1 &   0 &   1 &  0.66 &  -0.42  &  -0.14\\
0 & 1 &   1 &   0 &  0.99 &  -0.01  &   0.25\\
0 & 1 &   1 &   1 &  0.80 &  0.12   &   0.25\\
1/4 & 3/4 &   0 &   0 &  0.58 &  -0.81  & -0.27\\
1/4 & 3/4 &   0 &   1 &  0.87 &  -0.55  &  -0.43\\
1/4 & 3/4 &   1 &   0 &  0.48 &  -0.24  &  -0.16\\
1/4 & 3/4 &   1 &   1 &  0.89 &  -0.06  &  -0.08\\\hline
\end{tabular}
\end{center}
\caption{LPS of the probit stick-breaking, dependent Dirichlet process and NCoRM mixtures for Simulated Data Sets III.}
\label{sim_data_III}
\end{table}

The LPS for Simulated Data Sets III are shown in Table~\ref{sim_data_III}. For these data sets, the methods were ranked in the same order with the DDP giving the best performance and the NCoRM mixture outperforming the 
 probit stick-breaking mixtures. The difference between the DDP and NCoRM was largest for the model without jumps and with homoscedasticity ($c=0$ and $d=0$). This is not surprising since a Gaussian process with normal errors would provide good approximation of the sine curve and models which only allow dependence through the weights can only approximate the curve using piecewise constant fits. If there are jumps, the advantage of the DDP over the NCoRM was reduced.
 In all case, the NCoRM mixture substantially outperformed the probit stick-breaking mixture. 
 This reflects the construction of the probit stick-breaking processes.  In all stick-breaking processes, the weights are stochastically ordered and the probit stick-breaking process assumes that the atom with largest {\it a priori} expected weight does not  depend on $x$. Although the data can change the order {\it a posteriori}, this ordering persists in data sets of the size considered in these simulated examples.
 
The results of these simulations suggest some guidelines which can be used in more general situations. The model displayed in equation \eqref{DDP} is very general but in this form is rarely used in real situations. By allowing the parameter $\theta$ and the weights $w_k$ to depend on the regressor, the model becomes extremely flexible and prone to overfit the data. Therefore, the models considered are typically used as special cases of the model displayed in equation \eqref{DDP}. We prefer the NCoRM mixture model to the DDP mixture model if we can identify subpopulations with different levels of response and which are associated with different regressor values. Clearly, this is the case in simulated datasets I and II but also simulated dataset III when jumps are introduced ($c=1$). On the other hand, the DDP works well if there is a wide range of responses in each subpopulation. 
 
\section{Conclusions}

Normalized compound random measures are a large class of dependent nonparametric processes. The jumps of the processes are expressed as the product of a jump from a L\'evy process and a random variable. This allows the dependence of the nonparametric processes to be modelled through the dependence in the random variables. In this paper, we have developed Markov chain Monte Carlo methods to estimate nonparametric mixture models where the mixing measure is given a normalized compound random measure prior with a wide-range of dependences between the underlying random variables. The NCoRM approach could be generalized to allow the jump locations to depend on regressors and the MCMC method could be simply extended. The examples illustrate priors constructed using linear models and Gaussian processes. Other types of dependence could be included such as time series models, spatial models or hierarchical models. The MCMC methods are efficient and depend on approximating the tail mass integral of a L\'evy process. Examples of appropriate approximations are given for the most popular classes of L\'evy processes used in Bayesian nonparametrics. A simulation study illustrates that the NCoRM can provide better out-of-sample predictive performance than probit stick-breaking process mixtures in a range of simulated data sets and can outperform DDP mixtures if the mean of the responses does not vary smoothly.

\appendix
\section{Proofs}

\begin{proof}[Proof of Theorem \ref{PoissonThm}]
\[
\E\left[1-\frac{\phi(x_i)}{a\,C\,\kappa(x_i)}\right]
=1-\frac{1}{a\,C}\E\left[\frac{\phi(x_i)}{\kappa(x_i)}\right]
=1-\frac{1}{a\,C}\int_D \frac{\phi(x_i)}{\kappa(x_i)}\kappa(x_i)\,dx_i
=1-\frac{1}{a\,C}\int_D \phi(x_i)\,dx_i
\]
\begin{align*}
\E\left[\left(1-\frac{1}{a\,C}\frac{\phi(x_i)}{\kappa(x_i)}\right)^2\right]
&=
1 - 2\frac{1}{a\,C}\int_D \phi(x_i)\,dx_i
+\frac{1}{a^2\,C^2}\int_D \frac{\phi(x_i)^2}{\kappa(x_i)}\,dx_i.
\end{align*}
If $k\sim\Pn(\mu)$
\[
\E[(1-b)^k]=\exp\{-\mu b\}
\]
and so
\[
\E[\hat{L}_{\phi}]=\exp\left\{-\int_D \phi(x)\,dx\right\}
\]
\begin{align*}
\V[\hat{L}_{\phi}]&=\E[\hat{L}_{\phi}^2]-\E[\hat{L}_{\phi}]^2\\
&=\exp\left\{-
 2\int_D \phi(x)\,dx
+\frac{1}{a\,C}\int_D \frac{\phi(x)^2}{\kappa(x)}\,dx\right\}
-\exp\left\{-2\int_D \phi(x)\,dx\right\}\\
&=L_{\phi}^2
\left(\exp\left\{\frac{1}{a\,C}
\int_D \frac{\phi(x)^2}{\kappa(x)}\,dx\right\}-1\right)\\
\end{align*}
\end{proof}

\section{Additional details of computational methods}

\subsubsection*{Updating $J_k$}

Update $J_k$ from the full conditional distribution proportional to
\begin{align*}
 J_k^{n_k}\,\exp\left\{-\sum_{i=1}^n v_i\,  J_k \,m_{k,i}\right\}\,\nu_{\xi}^{\star}(J_k).
\end{align*}
This variable can be updated in closed form if $\nu_{\xi}^{\star}$ is the L\'evy density of a generalized gamma process or by an adaptive Metropolis-Hastings random walk \citep{atros} if the full conditional does not have closed form.

\subsubsection*{Updating $m_k$}

Update $m_k$ from the full conditional distribution proportional to
\begin{align*}
\prod_{i=1; c_i=k}^n m_{k,i}\,\exp\left\{-\sum_{i=1}^n v_i\,  J_k\, m_{k,i}\right\}\,h(m_k\vert \tau).
\end{align*}
This variable can be updated using an adaptive Metropolis-Hastings random walk \citep{atros} if the full conditional does not have closed form.

\subsubsection*{Updating $v_i$}

The full conditional distribution is 
\[
 \exp\left\{-v_i \prod_{k=1}^K  J_k\, m_{k,i}\right\}
\hat{L}.
\]
The parameter $v_i$ is updated using an interweaving step \citep{yumeng11}. The first part of the step
updates using an adaptive Metropolis-Hastings random walk \citep{atros}. A new value  $v'_i$ is proposed and a new estimate $\hat{L}'$ conditional on $v'_i$ is calculated. The proposed values $v'_i$ and $\hat{L}'$ are accepted with probability
\[
\min\left\{1,
\frac{\exp\left\{-v'_i \prod_{k=1}^K  J_k\, m_{k,i}\right\}
\hat{L'}}{\exp\left\{-v_i \prod_{k=1}^K  J_k\, m_{k,i}\right\}
\hat{L}}
\right\}.
\]
The second part of the step uses the re-parameterization 
$ \tilde{J}_i = v_1 J_i$ and $\tilde{v}_i =  \frac{v_i}{v_1}$ for $i>1$ which implies that $\tilde{J}_j\,\tilde{v}_i = J_j \,v_i$. The full conditional of $v_1$ (conditioning on $\tilde{v}_i$ and $\tilde{J}_i$) has density proportional to
\[
v_1^{-K}\ \left\{ \exp\left\{-v_1 \sum_{k=1}^K J_k \,m_{k,1}\right\}
 \,\nu_{\xi}^{\star}\left(\frac{\tilde{J}_k}{v_1}\right)\right\}\hat{L}.
\]
The parameter is updated using an adaptive Metropolis-Hastings random walk \citep{atros} where $L'$ is calculated conditional on $v_1'$ and $\tilde{v}_i$ for $i>1$. The proposed values are accepted with probability
\[
\min\left\{1,\frac{v_1'^{-K}\ \left\{ \exp\left\{-v_1' \sum_{k=1}^K J_k \,m_{k,1}\right\}
 \,\nu_{\xi}^{\star}\left(\frac{\tilde{J}_k}{v_1'}\right)\right\}\hat{L}'}{v_1^{-K}\ \left\{ \exp\left\{-v_1 \sum_{k=1}^K J_k \,m_{k,1}\right\}
 \,\nu_{\xi}^{\star}\left(\frac{\tilde{J}_k}{v_1}\right)\right\}\hat{L}}\right\}.
\]

\subsubsection*{Updating $\xi$}

The full conditional distributions of $M$ has density proportional to
\[
\hat{L}\,p(\xi)\,\prod_{k=1}^K\nu^{\star}_{\xi}(J_k).
\]
This parameter can be updated using an adaptive Metropolis-Hastings random walk \citep{atros}.

\subsubsection*{Updating $M$}

The full conditional distributions of $M$ has density proportional to
\[
p(M)\,\hat{L}\,M^K.
\]
This parameter can be updated using an adaptive Metropolis-Hastings random walk \citep{atros}.

\subsubsection*{Updating $\tau$}\label{CompDet}

The full conditional distribution of $\tau$ has density proportional to 
\[
p(\tau)\hat{L}\,\prod_{k=1}^K p(m_k\vert \tau)
\]
This parameter can be updated using an adaptive Metropolis-Hastings random walk \citep{atros}. We have found that 
an interweaved update \citep{yumeng11} can lead to much better mixing. If $m_k(x)$ is a Gaussian process and $\tau$ is the stationary variance then we can write $m_{k,i} = \sqrt\tau m^{(s)}_{k,i}$. The interweaved update uses an adaptive Metropolis-Hastings step where $\hat{L}$ is calculated for the proposed value $\tau'$ and the full conditional density is 
\begin{align*}
p(\tau)\,\hat{L}\,\tau^{n/2}\,\exp\left\{-\sqrt{\tau}\sum_{i=1}^n v_i\,  J_k\, m^{(s)}_{k,i}\right\}.
\end{align*}
\end{document}